\documentclass[11pt]{article}

\usepackage[T1]{fontenc}
\usepackage{graphicx}
\usepackage[utf8]{inputenc}
\usepackage{amsmath,mathrsfs,mathtools,amsthm}
\usepackage{amsmath,amsfonts,amssymb,epsfig}
\usepackage[usenames,dvipsnames]{color}

\usepackage{graphicx}
\usepackage{epstopdf}
\usepackage{yfonts}
\usepackage[T1]{fontenc}
\usepackage[all]{xy}
\usepackage{times}
\usepackage{hyperref}
\hypersetup{
    colorlinks=true,
    linkcolor=blue,
    citecolor=blue,
    urlcolor=blue
}

\usepackage{caption}
\usepackage[linesnumbered,ruled,vlined]{algorithm2e}

\newtheorem{Definition}{Definition}
\newtheorem{Theorem}{Theorem}

\newtheorem{Corollary}{Corollary}

\newtheorem{Definition and Notation}{Definition and Notation}
\newtheorem{Lemma}{Lemma}

\begin{document}

\author{
  Maryam Bajalan\thanks{Institute of Mathematics and Informatics, Bulgarian Academy of Sciences, Bl. 8, Acad. G. Bonchev Str., 1113, Sofia, Bulgaria. Emails: maryam.bajalan@math.bas.bg.} \and
  Peter Boyvalenkov\thanks{Institute of Mathematics and Informatics, Bulgarian Academy of Sciences, Bl. 8, Acad. G. Bonchev Str., 1113, Sofia, Bulgaria. Emails: peter@math.bas.bg (Corresponding author).}
  }

\title{On irredundant orthogonal arrays}

\date{}

\maketitle
	\begin{abstract}
	An orthogonal array (OA), denoted by $\text{OA}(M, n, q, t)$, is an $M \times n$ matrix over an alphabet of size $q$ such that every selection of $t$ columns contains each possible $t$-tuple exactly $\lambda=M / q^t$ times. An irredundant orthogonal array (IrOA) is an OA with the additional property that, in any selection of $n - t$ columns, all resulting rows are distinct. IrOAs were first introduced by Goyeneche and \.{Z}yczkowski in 2014 to construct $t$-uniform quantum states without redundant information. Beyond their quantum applications, we focus on IrOAs as a combinatorial and coding theory problem. An OA is an IrOA if and only if its minimum Hamming distance is at least $t + 1$. Using this characterization, we demonstrate that for any linear code, either the code itself or its Euclidean dual forms a linear IrOA, giving a huge source of IrOAs. In the special case of self-dual codes, both the code and its dual yield IrOAs. Moreover, we construct new families of linear IrOAs based on self-dual, Maximum Distance Separable (MDS), and MDS-self-dual codes. Finally, we establish bounds on the minimum distance and covering radius of IrOAs.

	\end{abstract}

%%%%%%%%%%%%%%%%%%%%%%%%%%%%%%%%%%%%%%%%%%%%%%%%%%%%%%%%%%%%%%%%%%%%%%%%%%%%%%%%%%%%%%%%%%%%%%%%%%%%%%%%%%%%%%%%%%%%%%%%%%%%%%%%%%%%%%%%%%%%%%%%%%%%
\section{Introduction}
Irredundant orthogonal arrays (IrOAs), introduced by Goyeneche and Życzkowski in 2014 \cite{GZ2014}, are orthogonal arrays $\text{OA}(M, n, q, t)$ with the additional property that all rows in every projection onto $n - t$ coordinates are pairwise distinct. This irredundancy property is essential to construct $t$-uniform quantum states that contain no redundant information. IrOAs have since been studied from both quantum and combinatorial perspectives. For example, the connection between IrOAs and $t$-uniform quantum states and multipartite entangled states has been further explored in \cite{GRMZ2018,GZ2014, Pang_2019, Zang_2022}. More recently, explicit constructions and characterizations of IrOAs have been studied in \cite{Chen_2023, Chen_2021}, which offer algebraic techniques and general conditions for their existence.

It is known that a $t$-uniform state of $n$ qudits can exist only if $t \leqslant \lfloor n/2 \rfloor$, which implies $n \geqslant 2t$; see \cite{GZ2014}. Consequently, we adopt the assumption $n \geqslant 2t$ throughout this paper. An important equivalence, stated in Lemma \ref{1} (Lemma 2.1 in \cite{Chen_2023}), is that an orthogonal array is irredundant if and only if its minimum Hamming distance is at least $t + 1$. In particular, we show that when $n = 2t$, irredundancy is equivalent to having minimum distance exactly $t + 1$.

Using the equivalence between irredundancy and minimum distance, we explore how linear codes and their duals can give rise to IrOAs. Theorem \ref{sd-thm} shows that a linear code $C$ yields an IrOA if and only if $d \geqslant d^\perp$, where $d$ and $d^\perp$ denote the minimum distances of $C$ and its Euclidean dual $C^\perp$, respectively. Likewise, $C^\perp$ yields an IrOA if and only if $d^\perp \geqslant d$. As a result, at least one of $C$ and $C^\perp$ always defines an IrOA, and both do so if and only if $d = d^\perp$. Based on this result, we construct new families of IrOAs from well-known classes of codes. In particular, we identify IrOAs derived from self-dual codes. Moreover, we study irredundant orthogonal arrays arising from maximum distance separable (MDS) codes, which satisfy the Singleton bound $d = n - k + 1$. We also consider IrOAs derived from MDS self-dual codes, whose parameters are uniquely determined by their length $n$, as shown in Theorem \ref{083}. Building on these results, we provide some explicit criteria for when classical code families, such as generalized Reed-Muller codes and generalized Reed-Solomon codes, give rise to irredundant orthogonal arrays.

In the last section of the paper, we turn to fundamental metric parameters of irredundant orthogonal arrays: their minimum distance and covering radius. We establish upper and lower bounds on the minimum Hamming distance of IrOAs when the index $\lambda$ is greater than $1$. Then we focus on the covering radius of IrOAs. In the extremal case $M = q^{n - t}$, corresponding to equality in the upper bound of the size of an IrOA, we prove that the covering radius is exactly $t$. 
%%%%%%%%%%%%%%%%%%%%%%%%%%%%%%%%%%%%%%%%%%%%%%%%%%%%%%%%%%%%%%%%%%%%%%%%%%%%%%%%%%%%%%%%%%%%%%%%%%%%%%%%%%%%%%%%%%%%%%%%%%%%%%%%%%%%%%%%%%%%%%%%%%%%
\section{Basic properties of irredundant orthogonal arrays}

We begin by recalling the definition of orthogonal arrays (OAs).
We remark that OAs are considered as sets, and at this point, we do not need any structure of the alphabet.

\begin{Definition} \label{def-oa}
Let $H_q^n$ denote the Hamming space of length $n$ over an alphabet $H_q$ of size $q$. An orthogonal array $\text{OA}(M, n, q, t)$ with index $\lambda = M/q^t$ is a set of $M$ codewords in $H_q^n$ such that every $M \times t$-subarray contains each of the $q^t$ possible $t$-tuples exactly $\lambda$ times.
\end{Definition}

Irredundant orthogonal arrays (IrOAs) were first introduced in \cite{GZ2014}. The concept was developed to ensure that an OA satisfies the necessary condition for constructing $t$-uniform quantum states; see \cite[Section IV]{GZ2014}.

\begin{Definition} \label{def-irr-oa} \cite{GZ2014}
An orthogonal array $\text{OA}(M, n, q, t)$ is called irredundant, denoted $\text{IrOA}(M, n, q, t)$, if every $M \times (n - t)$-subarray has pairwise distinct rows.
\end{Definition}

It follows immediately from Definition \ref{def-oa} that any $\text{OA}(M, n, q, t)$ is also an $\text{OA}(M, n, q, t_1)$ for every positive integer $t_1 < t$. Furthermore, by Definition \ref{def-irr-oa}, if an OA is irredundant for strength $t$, then it remains irredundant when viewed as an $\text{OA}(M, n, q, t_1)$, because $n - t_1 > n - t$, hence the condition of pairwise distinct rows is preserved. Therefore, we always assume that the parameter $t$ is maximal and refer to it as the strength of the OA.

Another immediate consequence of Definition \ref{def-irr-oa} is the inequality $M \leqslant q^{n - t}$, because otherwise some rows would necessarily repeat in any $M \times (n - t)$-subarray, a contradiction. Combined with $M = \lambda q^t \geqslant q^t$, this implies $n \geqslant 2t$. Consequently, we adopt the assumption $n \geqslant 2t$ throughout this paper.  This assumption aligns with the known constraint that a $t$-uniform state of $n$ qudits can exist only if $t \leqslant \lfloor n/2 \rfloor$, that is, $n \geqslant 2t$; see \cite{GZ2014}.

Orthogonal arrays and irredundant orthogonal arrays can also be viewed as codes in the Hamming space $H_q^n$. We adopt standard notation and conventions from coding theory. For an orthogonal array $C \subseteq H_q^n$, the minimum distance is defined by
    \begin{align*}
    d=d(C) := \min_{\substack{x, y \in C, \ x \neq y}} d(x, y),
    \end{align*}
where $d(x, y)$ denotes the Hamming distance between codewords $x$ and $y$. Moreover, the covering radius of $C$ is defined by
    \begin{align*}
\rho=\rho(C) := \max_{x \in H_q^n} \min_{y \in C} d(x, y),
    \end{align*}
representing the maximum Hamming distance from any vector in $H_q^n$ to the nearest codeword in $C$.

Further basic properties of irredundant orthogonal arrays are presented as statements now. 

\begin{Lemma}\cite[Theorem 4.21]{Hedayat_1999}\label{15}
The orthogonal array $\text{OA}(M, n, q, t)$ has index $\lambda = 1$ if and only if $d = n - t + 1$.
\end{Lemma}

\begin{Lemma}\cite[Theorem 2]{GZ2014}\label{046}
Every orthogonal array $\text{OA}(M, n, q, t)$ of index $\lambda = 1$ is irredundant.
\end{Lemma}

\begin{Lemma} \cite[Lemma 2.1]{Chen_2023} \label{1}
The orthogonal array $\text{OA}(M, n, q, t)$ is irredundant if and only if $d \geqslant t + 1$.
\end{Lemma}

\begin{proof} Let $d\geqslant t+1$, and assume that $\text{OA}(M, n, q, t)$ is not irredundant. Then there exist codewords $x, y$ in $\text{OA}(M, n, q, t)$ which coincide in $n-t$ positions. This means that $d(x,y) \leqslant n-(n-t)=t$, a contradiction. Thus, $\text{OA}(M, n, q, t)$ is irredundant. 

Let $\text{OA}(M, n, q, t)$ be irredundant and assume that $d \leqslant t$. Then there are two distinct codewords $x,y$ in $\text{OA}(M, n, q, t)$ that differ in at most $t$ positions. This means that $x$ and $y$ coincide in the remaining at least $n-t$ positions, a contradiction to the irredundancy condition. Thus, $d \geqslant t+1$.
\end{proof}

From the inequalities $M \leqslant q^{n - t}$ and $n \geqslant 2t$, it is clear that the cases $M = q^{n - t}$ and $n = 2t$  are extreme cases for IrOAs. These cases are described in the next two Theorems.  

\begin{Theorem}\label{047}
Let $n=2t$ and $C$ be an $\text{OA}(M, n, q, t)$. Then the following are equivalent: 
\begin{enumerate}
    \item[{\rm (1)}] $C$ is irredundant.
    \item[{\rm (2)}] the index of $C$ is $\lambda=1$.
    \item[{\rm (3)}] the minimum distance of $C$ is $d=t+1=n-t+1$.
\end{enumerate}
\end{Theorem}

\begin{proof}
By Lemma \ref{15} and the assumption $n=2t$, the statements $(2)$ and $(3)$ are equivalent.
If $\lambda = 1$, then by Lemma \ref{046}, $\text{OA}(M, n, q, t)$ is irredundant, which shows that $(2)$ implies $(1)$. Suppose that $\text{OA}(M, n, q, t)$ is irredundant. Then every $M \times (n - t)$-subarray must consist of distinct rows. Since $n - t = t$, it follows that every $t$-column subarray contains all $q^t$ possible $t$-tuples, each appearing exactly once.  Hence, $\lambda = 1$, showing that {\rm(1)} implies {\rm(2)}.
\end{proof}

\begin{Theorem}\label{035} 
An orthogonal array $\text{OA}(q^{n-t}, n, q, t)$ is irredundant if and only if $n = 2t$ (when also $\lambda=1$ and $d=t+1$). 
\end{Theorem}

\begin{proof} The inequality $M \leqslant q^{n-t}$ was already mentioned as an immediate consequence of Definition \ref{def-irr-oa}. If the equality $M = q^{n - t}$ holds, then each $M \times (n - t)$-subarray contains all possible $(n - t)$-tuples exactly once. This means the OA has strength $n - t$ and index $\lambda = 1$. Because the OA has maximal strength $t$ by assumption, we have $n - t \leqslant t$, and hence $ n = 2t $, which implies $\lambda=1$ and $d=t+1$. 

If $n=2t$, Theorem \ref{047} implies that $\lambda=1$; i.e., any $M \times t=M \times (n-t)$-subarray has all possible distinct rows of length $n-t$, which means $M = q^{n-t}$. 
\end{proof}
%%%%%%%%%%%%%%%%%%%%%%%%%%%%%%%%%%%%%%%%%%%%%%%%%%%%%%%%%%%%%%%%%%%%%%%%%%%%%%%%%%%%%%%%%%%%%%%%%%%%%%%%%%%%%%%%%%%%%%%%%%%%%%%%%%%%%%%%%%%%%%%%%%%%
\section{IrOAs from linear codes and their duals}

\subsection{Dual codes and IrOAs}
Throughout this section, let $q$ be a prime power. Let $C$ be a linear code of length $n$ over $\mathbb{F}_q$. The Euclidean dual code $C^\perp$ is defined by
\begin{align*}
C^\perp= \left\{ \mathbf{v} \in \mathbb{F}_q^n \,\, : \,\, \mathbf{v} \cdot \mathbf{c} = 0 \,\, \text{for all} \,\, \mathbf{c} \in C \right\},   
\end{align*}
where $\mathbf{u} \cdot \mathbf{v} = u_1v_1 + \cdots + u_nv_n$ is the standard inner product (here $u=(u_1,\ldots,u_n)$, $v=(v_1,\ldots,v_n)$, as usual). If $C$ has parameters $[n, k, d]_q$, then $C^\perp$ has parameters $[n, n - k, d^\perp]_q$, and $(C^\perp)^\perp = C$. The codewords of $C$ form the rows of a linear orthogonal array $\text{OA}(q^k, n, q, t)$ of strength $t = d^\perp - 1$ (cf. \cite[Theorem 4.6]{Hedayat_1999}). Similarly, $C^\perp$ forms a linear orthogonal array $\text{OA}(q^{n-k}, n, q, d - 1)$ of strength $t^\perp = d - 1$.

The well-known inequality $d + d^\perp \leqslant n + 2$ (see, e.g., \cite[Theorem 4.5]{Levenshtein_1998}) implies $n \geqslant 2\min\{t, t^\perp\}$, which suggests a fundamental connection between Euclidean duality and irredundant orthogonal arrays. As the following theorem shows, for any linear code $C$, at least one of $C$ and $C^\perp$ is an irredundant orthogonal array.

\begin{Theorem} \label{sd-thm}
Let $C = [n,k,d]_q$ be a linear code over $\mathbb{F}_q$, and let $C^\perp$ be its Euclidean dual. Then $n \geqslant 2 \min\{t, t^\perp\}$, and the following hold: \begin{enumerate}
    \item[{\rm (1)}] $C$ ($C^\perp$, resp.) is an irredundant OA if and only if $d \geqslant d^\perp$ ($d^\perp \geqslant d$, resp.). In particular, if $d \neq d^\perp$, then exactly one of the OAs $C$ and $C^\perp$ is an irredundant OA (the one with greater minimum distance).  
    \item[{\rm (2)}] Both $C$ and $C^\perp$ are irredundant OAs if and only if $d^\perp = d$.
 \end{enumerate}
\end{Theorem}

\begin{proof}
 The code $ C $ is an orthogonal array $ \text{OA}(q^k, n, q, t) $ of strength $ t = d^\perp - 1 $. By Lemma \ref{1}, $ C $ is irredundant if and only if $ d \geqslant t + 1 $, which is equivalent to $  d \geqslant (d^\perp - 1) + 1=d^\perp$. The proof for $C^{\perp}$ is similar by interchanging the roles of $C$ and $C^\perp$. The second statement follows directly from the first.
\end{proof}
%+++++++++++++++++++++++++++++++++++++++++++++++++++++++++++++++++++++++++++++++++++++++++++++++++++++++++++++++++++++++++++++++++
\subsection{Self-dual codes and IrOAs}\label{085}
A linear code $C$ over a finite field is said to be self-dual if $C = C^{\perp}$. Since $\dim(C) + \dim(C^{\perp}) = n$, it follows that $n = 2k$; in particular, $n$ must be even and $k = n/2$. The theory of self-dual codes is well developed and has been extensively studied; see, for example, \cite{Bou2021, Huffman_2010, NRS-book}.  

The following theorem establishes a connection between self-dual codes and irredundant orthogonal arrays.

\begin{Theorem}\label{059}
 Every self-dual $[n,n/2,d]_q$ linear code is an $\text{IrOA}(q^{n/2}, n, q, t)$ with minimum distance $d = t + 1$. 
\end{Theorem}

\begin{proof}
It follows immediately from Theorem \ref{sd-thm}. Since $d = d^\perp$, the strength satisfies $t = d^\perp - 1 = d - 1$.
\end{proof}
%+++++++++++++++++++++++++++++++++++++++++++++++++++++++++++++++++++++++++++++++++++++++++++++++++++++++++++++++++++++++++++++++++
\subsection{Reed-Muller codes and IrOAs}\label{093}

The quantum relation of the classical Reed-Muller (RM) codes is known since 1999 \cite{Ste1999}. The quantum RM codes are a candidate for implementing universal quantum computation. For investigations of certain properties of quantum RM codes via the properties of their classical counterparts we refer to \cite{BCHK}. RM codes are also closely related to the polar codes (cf. \cite{Ari2009,Ari2010,MHU2014}, which have advanced theoretical and practical applications.

We now establish a connection between one of the most prominent families of codes, the Reed-Muller codes, and irredundant orthogonal arrays. 

Let $m$ be a positive integer and $r$ an integer with $0 \leqslant r \leqslant m$. Define the matrix $G(0, m) = [1\ 1\ \cdots\ 1]$ and the matrix $G(m, m)$ as the identity matrix $I_{2^m}$. The binary Reed-Muller code $R(r, m)$ is a linear code of length $n = 2^m$, defined recursively by the following generator matrix
\begin{align*}
    G(r,m) =
\begin{bmatrix}
G(r,m - 1) & G(r,m - 1) \\
0 & G(r - 1,m - 1)
\end{bmatrix}.
\end{align*}
For example, the generator matrices for $R(1, 2)$ and $R(1, 3)$ are given by
\begin{align}\label{090}
    G(1, 2) =
\begin{bmatrix}
1 & 0 & 1 & 0 \\
0 & 1 & 0 & 1 \\
0 & 0 & 1 & 1
\end{bmatrix}, \quad
G(1, 3) =
\begin{bmatrix}
1 & 0 & 1 & 0 & 1 & 0 & 1 & 0 \\
0 & 1 & 0 & 1 & 0 & 1 & 0 & 1 \\
0 & 0 & 1 & 1 & 0 & 0 & 1 & 1 \\
0 & 0 & 0 & 0 & 1 & 1 & 1 & 1
\end{bmatrix}.
\end{align}

The main parameters of RM codes are given below.
\begin{Lemma}\cite[Theorem 1.10.1]{Huffman_2010}
    For any RM code $ R (r,m)$, the following hold:
    \begin{enumerate}
        \item [{\rm (1)}] The dimension of $R(r,m)$ is $k={\sum_{i=0}^{r} \binom{m}{i}}$.
        \item [{\rm (2)}] The minimum Hamming distance of $R(r,m)$ is $d=2^{m-r}$. 
        \item [{\rm (3)}] The dual of $R(r,m)$ is $R(r,m)^{\perp} = R(m - r - 1,m)$.
    \end{enumerate}
\end{Lemma}
Because  the dual of a 
Reed-Muller code is again a Reed-Muller code, it follows that $d^{\perp}=2^{m-(m-r-1)}=2^{r+1}$. Therefore, the code $R(r, m)$ is an $\text{OA}(2^k,\, 2^m,\, 2,\, 2^{r+1}-1)$; see \cite[Section 5.8]{Hedayat_1999}).

We are now prepared to relate RM codes to IrOAs.

\begin{Theorem}[IrOAs constructed from Reed-Muller codes]\label{070}\
\begin{enumerate}
    \item[{\rm (1)}] If $m=2r+1$, then $R(r,m)=R(r,m)^{\perp}$ is an $\text{IrOA}(2^{2^{m-1}}, \, 2^m, \, 2, \, 2^{r+1} - 1)$.  
    \item[{\rm (2)}] If $1 \leqslant r \leqslant m < 2r + 1$, then $R(r,m)^{\perp}$ is an $\text{IrOA}(2^{k'},\, 2^m,\, 2,\, 2^{m-r}-1)$, where $k'={\sum_{i=0}^{m-r-1} \binom{m}{i}}$.
    \item[{\rm (3)}] If $m > \max\{2r + 1,\, r + 2\}$, then $R(r,m)$ is an $\text{IrOA}(2^k, \, 2^m, \, 2, \, 2^{r+1}-1)$, where $k={\sum_{i=0}^{r} \binom{m}{i}}$.
\end{enumerate}
\end{Theorem}

\begin{proof}
 \begin{enumerate}
 \item[{\rm (1)}] A Reed-Muller code $R(r,m)$ is self-dual if and only if  $R(r,m)=R(m - r - 1,m)$, which  holds when $m=2r+1$. In this case, $R(r,m)$ is a binary self-dual code with parameters $[2^m, 2^{m-1}, 2^{r+1}]_2$. By Theorem \ref{059}, it is an $\text{IrOA}(2^{2^{m-1}}, \, 2^m, \, 2, \, 2^{r+1}-1)$. The necessary condition for irredundancy $n\geqslant 2t$ is equivalent to $ 2^m \geqslant 2 ( 2^{r+1} - 1)$, which is easily seen to hold true (with equality only for $m=1$). 
 
 \item[{\rm (2)}] The Reed-Muller code $R(r,m)^{\perp}$ is the orthogonal array $\text{OA}(2^{k'},\, 2^m,\, 2,\, 2^{m-r}-1)$, where $k'={\sum_{i=0}^{m-r-1} \binom{m}{i}}$. Under the condition $m< 2r+1$, we have $d^{\perp}> d$, and hence by Theorem \ref{sd-thm}(2), $R(r,m)^{\perp}$ is an irredundant orthogonal array. Note that the necessary condition $n \geqslant 2t$ for redundancy holds under the condition $r\geqslant 1$.

 \item [{\rm (3)}] Under the condition $m > 2r + 1$, the Reed-Muller code $R(r, m)$ satisfies $d > d^{\perp}$. Therefore, by Theorem \ref{sd-thm}(1), $R(r, m)=\text{OA}(2^k,, 2^m, 2, 2^{r+1} - 1)$ is irredundant. The necessary condition for redundancy holds under the condition $m\geqslant r+2$.
\end{enumerate}
\end{proof}
%+++++++++++++++++++++++++++++++++++++++++++++++++++++++++++++++++++++++++++++++++++++++++++++++++++++++++++++++++++++++++++++++++
\subsection{Generalized Reed-Muller codes and IrOAs}

Let $m$ be a positive integer and $r$ an integer with $0 \leqslant r \leqslant m(q - 1)$. Define \[ B_q (r,m) = \left\{ x_1^{e_1} \cdots x_m^{e_m} \;\middle|\; 0 \leqslant e_i < q,\, e_1 + e_2 + \cdots + e_m \leqslant r \right\} \]
as the set of monomials in $\mathbb{F}_q [x_1, x_2, \ldots, x_m]$ of total degree at most $r$, where the degree of each variable is less than $q$. The generalized Reed-Muller (GRM) code, denoted by $R_{q} (r,m)$, is a linear code of length $n = q^m$ defined as (see, e.g., \cite[Definition 5.3]{Assmus_1998} and \cite[Definition 1]{Ding_2000})
\begin{align}\label{095}
    R_{q} (r,m) = \left\langle x_1^{e_1} \cdots x_m^{e_m} \;:\; 0 \leqslant e_i < q,\, e_1 + e_2 + \cdots + e_m \leqslant r \right\rangle=\operatorname{Span}_{\mathbb F_q}(B_q(r, m)).
\end{align}
Let $P_1, P_2, \ldots, P_n$ be the $n = q^m$ points in the affine space $A^m(\mathbb{F}_q)=\mathbb F_q^m$. The codewords of a GRM code are obtained by evaluating polynomials in $B_q(r, m)$ at all the points $P_1, P_2, \ldots, P_n$ (see the explanation following \cite[Definition 1]{Ding_2000}). Thus, the code can be expressed as
\begin{align*}%\label{091}
    R_q (r,m) &= \operatorname{Span}_{\mathbb F_q} \left\{ ( f (P_1), f (P_2), \ldots, f (P_n)) \,:\, f \in B_q (r,m) \right\}\\
    &= \left\{ \left(f(P_1), f(P_2), \ldots, f(P_n)\right) \,:\, f \in \operatorname{Span}_{\mathbb{F}_q}(B_q(r,m)) \right\}.
\end{align*}
As an example, let $ q=2 $, $ m=2 $, so the evaluation points are $ (0,0), (0,1), (1,0), (1,1) $. If $r=1$, we have $B_2 (1,2) =\{1,x_1,x_2\}$. Therefore, one generator matrix of $R_2(1,2)$ is
\begin{align*}
\begin{bmatrix}
  1 & 1 & 1 & 1 \\
  0 & 0 & 1 & 1 \\
  0 & 1 & 0 & 1 \\
  \end{bmatrix}.
\end{align*}
This matrix is row-equivalent to the recursively defined generator matrix $G(1,2)$ given in \eqref{090}. To see that GRM codes over $\mathbb{F}_2$ are the same as the recursively defined Reed–Muller codes, see \cite[Exercise 772]{Huffman_2010}.

The main parameters of the GRM codes follow. 

\begin{Lemma}\cite[Theorems 5.5 and 5.8 and Corollary 5.26]{Assmus_1998}\label{073}
For any GRM code $ R_q (r,m)$, we have 
    \begin{enumerate}
        \item[{\rm (1)}] $\dim(R_{q} (r,m)) = \sum_{i=0}^{m} (-1)^i \binom{m}{i} \binom{m + r - iq}{r - iq}$.
        \item [{\rm (2)}] $R_{q} (r,m)^\perp = R_{q} (m(q - 1) - 1 - r,\, m)$.
        \item [{\rm (3)}]  $R_{q} (r,m)$ has minimum distance $d=(q-b)q^{m-a-1}$, where $a$ and $b$ are the unique integers, defined by $r = a(q - 1) + b$ and $0 \leqslant b < q - 1$.
    \end{enumerate}
\end{Lemma}

We are now in a position to relate the GRM codes and the IrOAs.

\begin{Theorem}[IrOAs constructed from generalized Reed-Muller codes] \label{thm-grm}
Let $r = a(q - 1) + b$, where $a$ and $b$ are the unique integers, defined by $r = a(q - 1) + b$ and $0 \leqslant b < q - 1$, and $r'=m(q-1)-1-r$.
    \begin{enumerate}
        \item [{\rm (1)}] If $m(q-1)=2r+1$, then $R_q(r,m)=R_q(r,m)^{\perp}$ is an $\text{IrOA}(q^{q^{m-1}}, \, q^m, \, q, \, d - 1)$, where $d=(q-b)q^{m-a-1}$.
        \item [{\rm (2)}] If $ m(q - 1) < 2r + 1 $ and $q^{a + 1} \geqslant 2(q - b)$, then $R_q(r,m)^{\perp}$ is an $\text{IrOA}(q^k, \, q^m, \, q, \, d - 1)$, where $d=(q-b)q^{m-a-1}$  and $k= \sum_{i=0}^{m} (-1)^i \binom{m}{i} \binom{m + r' - iq}{r' - iq}$.
        \item [{\rm (3)}] If $ m(q - 1) > 2r + 1 $ and $ q^{m-a} \geqslant 2(b + 2) $, then $ R_q(r, m) $ is an $ \text{IrOA}(q^k, \, q^m, \, q, \, d^\perp - 1) $, where $ d^\perp = (b + 2)q^{a} $ and $ k = \sum_{i=0}^{m} (-1)^i \binom{m}{i} \binom{m + r - iq}{r - iq} $.    
    \end{enumerate}
\end{Theorem}

\begin{proof}
    \begin{enumerate}
        \item [{\rm (1)}] The proof is similar to the proof of Theorem \ref{070}(1). We only need to verify the necessary condition for irredundancy, namely $n \geqslant 2t$, or equivalently, $q^m\geqslant 2\left((q-b)q^{m-a-1}-1\right)$.  Substituting $ r = a(q - 1) + b $ into the assumtion $m(q-1)=2r+1$ implies $(m - 2a)(q - 1) = 2b + 1$. We have $(2b+1)/(q-1) \in \big[1/(q-1),2+1/(q-1)\big)$ since $ 0 \leqslant b < q - 1 $. The case $(2b+1)/(q-1) \in \big[2,2+1/(q-1)\big)$ cannot occur, because an even number cannot lie between two consecutive odd numbers. Therefore,  $(2b+1)/(q-1) \in \big[1/(q-1),2\big)$, and hence the integer number $m-2a$ lies in the interval $ \big[1/(q-1),2\big)$. Thus, $m-2a=1$, which leads to $2b+1=q-1$. Plugging $m = 2a + 1$ into the inequality $q^m \geqslant 2\left((q - b)q^{m - a - 1} - 1\right)$, we obtain $q^{2a + 1} \geqslant (q + 2)q^a - 2$, which always holds.

        \item [{\rm (2)}] Applying $ r = a(q - 1) + b $, we get $ r'= (m - a - 1)(q - 1) + (q - b - 2)$. Define $ a' := m - a - 1 $ and $ b' := q - b - 2 $. Since $ 0 \leqslant b < q - 1 $, we obtain $ 0 \leqslant b' < q - 1 $.
        Applying Lemma \ref{073} with these values of $a'$, $b'$, $r'$, we obtain
        \begin{align}\label{076}
        d = (q - b)q^{m-a-1},\,\,\, d^\perp = (b + 2)q^{a}.  
        \end{align}\label{078}
        Substituting $ r = a(q - 1) + b $ in the assumption $ m(q - 1) < 2r + 1 $ implies $(m - 2a)(q - 1) < 2b + 1$. Applying the interval $(2b+1)/(q-1) \in \big[1/(q-1),2\big)$ from the proof of the first statement, we obtain $m-2a<(2b+1)/(q-1)<2$. Since $m-2a$ is an integer, we have $m-2a\leqslant 1$.
        If $m=2a+1$, then applying $(m - 2a)(q - 1) < 2b + 1$, we get $q - b < b + 2$. Thus applying \ref{076}, $d^\perp>d$. Finally, if $m \leqslant 2a$, then 
        \[ \frac{d^\perp}{d} \geqslant \frac{b+2}{q-b} \cdot q =\frac{q(b+2)}{q-b} \geqslant b+2>1. \]

        Hence, by Theorem \ref{sd-thm}, $R_q(m,r)^\perp=\text{OA}(2^{k'},\, 2^m,\, 2,\, 2^{m-r}-1)$ is irredundant. To complete the proof, we need to verify the necessary condition for irredundancy, which is true by the assumption $q^{a + 1} \geqslant 2(q - b)$. 
        
        \item [{\rm (3)}] Substituting $ r = a(q - 1) + b $ in the assumption $ m(q - 1) > 2r + 1 $ implies $(m - 2a)(q - 1) > 2b + 1$. Therefore, $m-2a>(2b + 1)/(q-1)\geqslant 1/(q-1)>0$, and since $m-2a$ is integer, we have $m-2a\geqslant 1$. We conclude that $m = 2a + 1 + t$ for some integer $t \geqslant 0$.
Substituting $m = 2a + 1 + t$ in \eqref{076}, 
\begin{align*}
    \frac{d}{d^\perp} \geqslant \frac{q - b}{b + 2}\cdot q^t.
\end{align*}
If $t = 0$, then applying $(m - 2a)(q - 1) > 2b + 1$, we get $q - b > b + 2$. Thus $d>d^{\perp}$. Finally, if $t \geqslant 1$, then applying $0 \leqslant b < q - 1$, we derive $q - b \geqslant 2$, $b + 2 < q + 1$, and $(q - b)q^{t} \geqslant 2q>q + 1 > b + 2$. Thus $d>d^{\perp}$. Therefore, by Theorem \ref{sd-thm}, $R_q(m,r)= \text{OA}(q^k, \, q^m, \, q, \, d^\perp - 1) $ is irredundant. The necessary condition for redundancy holds using the assumption $ q^{m-a} \geqslant 2(b + 2) $.
    \end{enumerate}
\end{proof}

We remark that the exceptions for the assumptions $q^{a+1}>2(q-b)$ and $q^{m-a}>2(b+2)$ in Theorem~\ref{thm-grm}(2) and (3), respectively, lead to small trivial cases, making the statement similar to the binary case. 
%+++++++++++++++++++++++++++++++++++++++++++++++++++++++++++++++++++++++++++++++++++++++++++++++++++++++++++++++++++++++++++++++++
\subsection{MDS codes and IrOAs}

Recall that the Singleton bound for an $[n, k, d]_q$ code states that $d \leqslant n - k + 1$. A code that meets this bound, i.e., $d = n - k + 1$, is called a Maximum Distance Separable (MDS) code. For a recent survey on MDS we refer to \cite{Ball2020}. There are limitations on the existence of MDS codes and we refer to the well-known MDS Conjecture (related to a question from 1955) which still remains unsolved in some cases.

MDS Conjecture: If $k \leqslant q$, then a linear $[n, k, n - k + 1]_q$ MDS code exists exactly when $n \leqslant q + 1$, unless $q = 2^h$ and $k = 3$ or $k = q - 1$, in which case it exists exactly when $n \leqslant q + 2$; see, e.g., \cite[Conjecture 3.3.21]{Bou2021}.

A linear code $C$ over a finite field is MDS if and only if its dual $C^\perp$ is MDS (see, e.g., \cite[Theorem 2.4.3]{Huffman_2010}). This paves the way for IrOAs relationships. 

\begin{Theorem}\label{063}
    Let $C$ be an MDS $[n, k, n-k+1]_q$ linear code. Then
    \begin{enumerate}
        \item[{\rm (1)}] $C$ is an $\text{OA}(q^{k}, n, q, k)$ with minimum distance $d=n-k+1$.
        \item[{\rm (2)}] $C^{\perp}$ is an $\text{OA}(q^{n-k}, n, q, n-k)$ with minimum distance $d^\perp=k+1$.
        \item[{\rm (3)}] $C$ is irredundant if and only if $n\geqslant 2k$. 
        \item[{\rm (4)}] $C^\perp$ is irredundant if and only if $n\leqslant 2k$.  
    \end{enumerate}
\end{Theorem}

\begin{proof}
As already said several times, $C$ is an $\text{OA}(q^{k}, n, q, t)$ with $t=d^\perp-1$. Since both $C$ and $C^\perp$ are MDS, we have $d=n-k+1$ and $d^{\perp}=k+1$, respectively, for their minimum distances. Therefore, $t=d^{\perp}-1=k$ and $t^\perp=n-k$, respectively, which proves (1) and (2). The equalities $t=k$ and $t^\perp=n-k$ and Lemma \ref{046} imply (3) and (4), respectively. 
\end{proof}
As observed, any self-dual linear code over $\mathbb{F}_q$ with parameters $[n, k, d]_q$ must satisfy $k = \tfrac{n}{2}$. Furthermore, for any MDS code over $\mathbb{F}_q$, the minimum distance satisfies $d = n - k + 1$. Combining these two facts, it follows that the parameters of a self-dual MDS code are entirely determined by its length $n$, and take the form $[n, \tfrac{n}{2}, \tfrac{n}{2} + 1]_q$. The following theorem follows immediately from Theorem~\ref{063}.

\begin{Theorem}\label{083}
 Every MDS self-dual $[n, \tfrac{n}{2}, \tfrac{n}{2} + 1]_q$ linear code is an $\text{IrOA}(q^{n/2}, n, q,\tfrac{n}{2})$.
\end{Theorem}
%+++++++++++++++++++++++++++++++++++++++++++++++++++++++++++++++++++++++++++++++++++++++++++++++++++++++++++++++++++++++++++++++++
\subsection{Generalized Reed-Solomon codes and IrOAs}

Generalized Reed-Solomon codes (GRS) are an important class of MDS codes. One description is the following. Suppose that $n$ is an integer satisfying $1 \leqslant n \leqslant q$, and let $a = (\alpha_1, \alpha_2, \dots, \alpha_n)$ be an $n$-tuple of distinct elements of $\mathbb{F}_q$, and $v = (v_1, v_2, \dots, v_n)$ be an $n$-tuple of non-zero (not necessarily distinct) elements of $\mathbb{F}_q$. The generalized Reed–Solomon (GRS) code of dimension $k$ associated with $a$ and $v$ is defined by
\begin{align*}
    \text{GRS}_k(a, v) := \{(v_1f(\alpha_1), v_2f(\alpha_2), \dots, v_nf(\alpha_n)) : f(x) \in \mathbb{F}_q[x],\, \deg f \leqslant k - 1\}.
\end{align*}
When $v_i = 1$ for all $i$, the code is called a Reed–Solomon (RS) code. One well-known generator matrix of $\mathrm{GRS}_k(a, v)$ is
\begin{align*}
    G =
\begin{bmatrix}
v_1 & v_2 & \cdots & v_{n} \\
v_1\alpha_1 & v_2\alpha_2 & \cdots & v_{n}\alpha_{n-1} \\
v_1\alpha_1^2 & v_2\alpha_2^2 & \cdots & v_{n}\alpha_{n-1}^2 \\
\vdots & \vdots & & \vdots \\
v_1\alpha_1^{k-1} & v_2\alpha_2^{k-1} & \cdots & v_{n}\alpha_{n-1}^{k-1}
\end{bmatrix}.
\end{align*}
The code $\mathrm{GRS}_k(a, v)$ is an $[n, k, n - k + 1]$ MDS code (see, e.g., \cite[Chapter 5]{Huffman_2010}). Its dual is given by $\text{GRS}_k(a, v)^\perp =\text{GRS}_{n-k}(a, u)$ for some $u$; see \cite[Theorem 4, Chapter 10]{MacWilliams_1977} and \cite[Theorem 5.3.3]{Huffman_2010}.
%where $u= (u_1, \dots, u_n)$ with $u_i^{-1} = v_i \prod_{j \neq i} (\alpha_i - \alpha_j)$ \cite[Theorem 5.1.6]{...}.
Therefore, the relation between GRS codes and IrOAs is given as follows.

\begin{Corollary}
Let $1 \leqslant n \leqslant q$. Then
\begin{enumerate}
    \item[{\rm (1)}] $\text{GRS}_k(a, v)$ is an $\text{IrOA}(q^{k}, n, q, k)$ if $n \geqslant 2k$. 
     \item[{\rm (2)}] $\text{GRS}_k(a, v)^{\perp}$ is an $\text{IrOA}(q^{n-k}, n, q, n-k)$ if $n \leqslant 2k$.
\end{enumerate}
\end{Corollary}

\begin{Corollary}
Suppose that $ n $ is even and $n\leqslant q$. If $ n $ and $ q $ satisfy one of the following two conditions, then $ \mathrm{GRS}_k(a, v) $ is an $ \mathrm{IrOA}(q^{n/2}, n, q, \tfrac{n}{2}) $.
\begin{enumerate}
    \item [{\rm (1)}] $q$ is even.
    \item [{\rm (2)}] $q= 1 \pmod 4$, and $q \geqslant 4^n\times n^2$.
\end{enumerate}    
\end{Corollary}
\begin{proof}
If either condition holds, then by \cite[Theorem 3.2]{Jin_2016}, the code $ \mathrm{GRS}_k(a, v) $ is MDS self-dual. The result then follows from Theorem~\ref{083}.
\end{proof}
%%%%%%%%%%%%%%%%%%%%%%%%%%%%%%%%%%%%%%%%%%%%%%%%%%%%%%%%%%%%%%%%%%%%%%%%%%%%%%%%%%%%%%%%%%%%%%%%%%%%%%%%%%%%%%%%%%%%%%%%%%%%%%%%%%%%%%
\section{Bounds on minimum distance and covering radius of IrOAs}

\subsection{Bounds on the minimum distance}

Since an OA$(M,n,q,t)$ has index $\lambda=1$ if and only if 
its minimum distance is $d=n-t+1$, it is straightforward to conclude that $d \leqslant n-t$ whenever $\lambda>1$. One next step 
was done in \cite{Boyvalenkov_2022}.

\begin{Theorem}\cite[Theorem IV.5]{Boyvalenkov_2022}\label{3}
Suppose that $\text{OA}(M, n, q, t)$ has index $\lambda > 1$
and $n-t>\lambda(q-1)/(\lambda-1)$. Then $d \leqslant n - t-1$.
\end{Theorem}

\begin{Corollary}\label{00}
For any $\text{IrOA}(M, n, q, t)$ with index $\lambda > 1$, we have
\begin{enumerate}
    \item[{\rm (1)}]  $t+1\leqslant   d \leqslant n - t$.
    \item[{\rm (2)}]  $t+1\leqslant   d\leqslant n - t-1$ if $n-t>\frac{\lambda(q-1)}{\lambda-1}$ and $n\geqslant 2t+2$.
\end{enumerate}
\end{Corollary}

\begin{proof}
Both statements follow from Lemma \ref{1} and Theorem \ref{3}.
\end{proof}

Compared to the bounds in Corollary \ref{00}, we provide a tighter upper bound on the minimum distance of IrOAs for a sufficiently large index $\lambda$, particularly when $\lambda \geqslant q^m$ for some $m > 2$.

\begin{Theorem}\label{05}
Suppose that $\text{IrOA}(M, n, q, t)$ has index $\lambda > 1$. If $\lambda \geqslant q^m$ and $n \geqslant 2t + m$ for some integer $m \geqslant 1$, then
\begin{align*}
  d \leqslant n - t - m + 1.
\end{align*}
\end{Theorem}

\begin{proof}
Since $M=\lambda q^t \geqslant q^{t+m}$, it follows that $\log_q M \geqslant t + m$. Applying the Singleton bound gives
\begin{align*}
d \leqslant n - \lfloor \log_q M \rfloor + 1 \leqslant  n - t - m + 1.
\end{align*}
Note that the condition $n \geqslant 2t + m$ is necessary since we have to comply with the bound $d \geqslant t + 1$ by Lemma \ref{1}.
\end{proof}
%%%%%%%%%%%%%%%%%%%%%%%%%%%%%%%%%%%%%%%%%%%%%%%%%%%%%%%%%%%%%%%%%%%%%%%%%%%%%%%%%%%%%%%%%%%%%%%%%%%%%%%%%%%%%%%%%%%%%%%%%%%%%%%%%%%%%%

\subsection{Bounds on covering radius}

We recall that the covering radius of a code $C \subseteq H_q^n$ is 
\begin{align*}
    \rho=\rho(C) = \max_{x \in H_q^n} \min_{y \in C} d(x, y).
\end{align*}
Regarding the upper bounds for orthogonal arrays, Delsarte bound (see \cite{Delsarte1_1973,Delsarte2_1973}; cf. \cite{Cohen_1985} for a survey) states that $\rho \leqslant n-t$. Some conditional refinements were obtained in \cite{BRS2020}, namely $\rho \leqslant n-t-1$ if $n>t+q-1$ and $\rho \leqslant n-t-2$ if $n>2(t+q-1)$. Note that the Delsarte bound is attained by the RM codes for $m-3 \leqslant r \leqslant m$ and by the RS codes.  For further results on the covering radius of RM, GRM, RS, and GRS codes, see, e.g., \cite{Bartoli_2014, Hou_1993, Leducq_2012}.

If $C \subset H_q^n$ is a code with covering radius $\rho$, then the sphere covering argument gives 
\begin{equation} \label{cov-arg}
|C| \cdot V_q(t, \rho) \geqslant q^{n}, 
\end{equation}
where $V_q(t, \rho) = \sum_{i=0}^\rho \binom{t}{i}(q-1)^i$ is the volume of a (Hamming) ball in $H_q^n$ of radius $\rho$.

The inequality \eqref{cov-arg} and the OAs and IrOAs properties imply the following restrictions. 

\begin{Theorem}\label{099}
We have $\lambda V_q(t,\rho) \geqslant q^{n-t}$ for any $\text{OA}(M, n, q, t)$ and $ V_q(t, \rho) \geqslant q^{t}$ for any $\text{IrOA}(M, n, q, t)$.
\end{Theorem}

\begin{proof}
For the first inequality, it is enough to use $M=\lambda q^t$.
The second inequality follows from $M\leqslant q^{n-t}$, which is true for any IrOA. 
\end{proof}

In the special case $M=q^{n-t}$, we detemine $\rho$.

\begin{Theorem}
The covering radius of any $\text{IrOA}(q^{n-t}, n, q, t)$ is $\rho= t$.    
\end{Theorem}

\begin{proof}
Let $C$ be an $\text{IrOA}(q^{n-t}, n, q, t)$ and $x \in H_q^n \setminus C$. Fix $n-t$ columns in the matrix formed by the codewords of $C$. Since every $(n-t)$-tuple appears exactly once as a row defined by the chosen $n-t$ columns, the point $x$ coincides with a row of $C$ in these $n-t$ columns. Therefore, $x$ and that row of $C$ are different in at most $t$ positions, i.e. $d(x,C) \leqslant t$, whence $\rho \leqslant t$ from the definition of covering radius.

Assume that $\rho<t$. By the binomial theorem, we obtain 
\begin{align*}
     V_q(t, \rho) = \sum_{i=0}^\rho \binom{t}{i}(q-1)^i < \sum_{i=0}^t \binom{t}{i}(q-1)^i=(1+(q-1))^t=q^t,
\end{align*}
which contradicts the second inequality of Theorem \ref{099}. Therefore, $\rho=t$.
\end{proof}

\begin{Theorem}
The covering radius of any $\text{IrOA}(M, n, q, t)$ satisfies the inequality $ \lfloor t/2\rfloor\leqslant \rho$.
\end{Theorem}

\begin{proof}
It is a well-known result (see \cite[Section 1.12]{Huffman_2010}) that the covering radius of any code $C$ satisfies 
\begin{align*}
    \left\lfloor \frac{d(C)-1}{2}\right\rfloor\leqslant \rho(C),
\end{align*}
where $d(C)$ is the minimum distance. By Lemma \ref{1}, $t+1\leqslant d(C)$, which completes the proof.
\end{proof}

%%%%%%%%%%%%%%%%%%%%%%%%%%%%%%%%%%%%%%%%%%%%%%%%%%%%%%%%%%%%%%%%%%%%%%%%%%%%%%%%%%%%%%%%%%%%%%%%%%%%%%%%%%%%%%%%%%%%%%%%%%%%%%%%%%
\section*{Acknowledgement} This research was supported by the project IC-TR/10/2024-2025. The authors would like to thank Alexander Barg for his valuable suggestions and helpful comments.

%%%%%%%%%%%%%%%%%%%%%%%%%%%%%%%%%%%%%%%%%%%%%%%%%%%%%%%%%%%%%%%%%%%%%%%%%%%%%%%%%%%%%%%%%%%%%%%%%%%%%%%%%%%%%%%%%%%%

\end{document}